\documentclass[aps,pra,twocolumn,superscriptaddress,nofootinbib,longbibliography]{revtex4-2}
\usepackage{amsmath}
\usepackage{makecell}
\usepackage{graphicx}
\usepackage{multirow}
\usepackage{tikz}
\usepackage{ragged2e} 
\usepackage[colorlinks=true, allcolors=blue]{hyperref}
\usepackage{booktabs}
\usepackage{adjustbox}
\usepackage{url}
\usepackage{array}        
\usepackage{algorithm}
\usepackage{algpseudocode}
\usepackage{amssymb}
\usepackage{braket}
\usepackage{bm}
\usepackage{amsthm}
\usepackage{color}
\usepackage{float}
\usepackage[left]{lineno}

\newtheorem{theorem}{Theorem}
\definecolor{Red}{rgb}{1,0,0}

\begin{document}


\title{Enhancing the reachability of variational quantum algorithms via input-state design}

\author{Shaojun Wu}
\affiliation{Institute of  Fundamental and Frontier Sciences, University of Electronic Science and Technology of China, Chengdu, Sichuan, 611731, China}
\affiliation{Key Laboratory of Quantum Physics and Photonic Quantum Information, Ministry of Education,
 University of Electronic Science and Technology of China, Chengdu 611731, China}

\author{Shan Jin}
\affiliation{Institute of  Fundamental and Frontier Sciences, University of Electronic Science and Technology of China, Chengdu, Sichuan, 611731, China}
\affiliation{Key Laboratory of Quantum Physics and Photonic Quantum Information, Ministry of Education,
 University of Electronic Science and Technology of China, Chengdu 611731, China}

\author{Abolfazl Bayat}
\email{abolfazl.bayat@uestc.edu.cn}
\affiliation{Institute of  Fundamental and Frontier Sciences, University of Electronic Science and Technology of China, Chengdu, Sichuan, 611731, China}
\affiliation{Key Laboratory of Quantum Physics and Photonic Quantum Information, Ministry of Education,
 University of Electronic Science and Technology of China, Chengdu 611731, China}
\affiliation{Shimmer Center, Tianfu Jiangxi Laboratory, Chengdu 641419, China}

\author{Xiaoting Wang}
\email{xiaoting@uestc.edu.cn}
\affiliation{Institute of Fundamental and Frontier Sciences, University of Electronic Science and Technology of China, Chengdu, Sichuan, 611731, China}
\affiliation{Key Laboratory of Quantum Physics and Photonic Quantum Information, Ministry of Education,
 University of Electronic Science and Technology of China, Chengdu 611731, China}

\date{\today}

\begin{abstract}

Variational quantum algorithms (VQAs) face an inherent trade-off between expressivity and trainability: deeper circuits can represent richer states but suffer from noise accumulation and barren plateaus, while shallow circuits remain trainable and implementable but lack expressive power. Here, we propose a general framework to address this challenge by enhancing the VQA performance with a specially designed input state constructed using a linear combination technique. This approach systematically modified the set of states reachable by the original circuit, enhancing accuracy while preserving efficiency. We provide a rigorous proof that such framework increases the expressive capacity of any given VQA ansatz, and demonstrate its broad applicability across different ansatz families. As applications, we apply the method to ground-state preparation of the transverse-field Ising, cluster-Ising, and Fermi-Hubbard models, achieving consistently higher accuracy under the same gate budget compared with standard VQAs. These results highlight input-state design as a powerful complement to circuit design in realizing VQAs that are both expressive and trainable.

\end{abstract}


\maketitle
\section{Introduction}

Quantum computers open a new direction towards computational problems which are inherently hard to solve by classical computers. While realization of large-scale fault-tolerant quantum computers is at least a decade away, the near-term quantum computers suffer from various imperfections in their initialization, operation, and readout. A crucial question is whether such Noisy Intermediate-Scale Quantum (NISQ) devices can achieve quantum advantage with respect to classical computers in solving practical problems~\cite{Preskill2018quantumcomputingin,RevModPhys.94.015004,TILLY20221,Fedorov2022}. Among the leading approaches in this domain are variational quantum algorithms (VQAs)~\cite{Cerezo2021Variational,Cerezo2022Challenge}, which formulate computational tasks as variational optimization problems where the target quantum state minimizes a cost function. In VQAs and their alternatives, such as quantum assisted simulation~\cite{PhysRevA.104.042418,PhysRevA.104.L050401,Haug_2022Generalized}, computational complexity is divided between a NISQ computer and a classical optimizer so that such hybrid structure may outperform purely classical computers. To date, VQAs have been successfully applied across a broad range of domains, including quantum machine learning~\cite{Biamonte2017QML,arunachalam2017surveyquantumlearningtheory,Dunjko_2018,farhi2018classificationquantumneuralnetworks,PhysRevLett.122.040504,Cong2019QCNN,Pesah2021Absence,Ren2022Experimental,PhysRevA.98.032309,PhysRevA.103.052414}, quantum chemistry~\cite{Peruzzo2014,Kandala2017,Arute2020}, quantum many-body problems~\cite{PRXQuantum.2.030301,sheils2024neartermquantumspinsimulation, PhysRevB.106.214429, PRXQuantum.3.010346, PhysRevB.105.094409, PhysRevB.106.144426,kokail2019self}, optimization tasks~\cite{farhi2014quantumapproximateoptimizationalgorithm,Lacroix2020,Harrigan2021,BLEKOS20241,farhi2015quantumapproximateoptimizationalgorithm,PhysRevX.10.021067,PhysRevA.97.022304,crooks2018performancequantumapproximateoptimization,PhysRevA.103.042612,Wurtz2022counterdiabaticity,farhi2020quantumapproximateoptimizationalgorithm,doi:10.1126/science.abo6587}, quantum metrology~\cite{Koczor2020VariationalState,meyer2021variational,Kaubruegger2021quantum,Kaubruegger2023Optimal} and quantum simulation~\cite{PhysRevX.7.021050,Yuan2019theoryofvariational,PhysRevA.99.062304,McArdle2019,PhysRevLett.125.010501,XU20212181Variational}.

In a general VQA, the objective is to approximate a target quantum state $\ket{\Psi_\mathrm{tar}}$ that minimizes the average of an observable $H$, namely $\bra{\Psi_\mathrm{tar}}H\ket{\Psi_\mathrm{tar}}$. In order to find the solution, one can engineer a set of parameterized quantum states $\ket{\Psi(\bm{\theta})}\equiv U(\bm{\theta})|\Psi_0\rangle$, where $U(\bm{\theta})$ is a unitary circuit with tunable parameters $\bm{\theta}$ and $\ket{\Psi_0}$ is an input state. If the ansatz state $\ket{\Psi(\bm{\theta})}$ is sufficiently expressive, there exists an optimal parameter set $\bm{\theta}_{\mathrm{opt}}$ such that $\ket{\Psi(\bm{\theta}_{\mathrm{opt}})}$ closely approximates $\ket{\Psi_{\mathrm{tar}}}$. This parameter set is determined by minimizing the cost function $E(\bm{\theta})=\langle \Psi_0|\,U^\dagger(\bm{\theta})\,H\,U(\bm{\theta})\,|\Psi_0\rangle$ through a classical optimization loop, thereby variationally approaching the ground energy of $H$. Generalizations of VQA targeting multiple eigenvectors have also been developed~\cite{PhysRevResearch.1.033062,PhysRevResearch.6.013015,hong2024refiningweightedsubspacesearchvariational,Ding2024groundexcitedstates}. However, if the target state $\ket{\Psi_{\mathrm{tar}}}$ is not reachable by any choice of $\bm{\theta}$, then $\ket{\Psi(\bm{\theta}_{\mathrm{opt}})}$ will not be the desired target state. The reachable set of the ansatz state $\ket{\Psi(\bm{\theta})}$ depends on both the unitary operator $U(\bm \theta)$ and the input state $\ket{\Psi_0}$. A central challenge in VQAs is to balance the trade-off between expressivity and trainability~\cite{McClean2018,Wang2021NoiseInduced,Volkoff_2021,Sim2020Expressibility}. Increasing the expressivity, e.g. by using deeper quantum circuits, can cover a larger reachable quantum states but suffers from noise accumulation and barren plateaus which prevents convergence in the classical optimization~\cite{McClean2018, Holmes2022Connecting,Larocca2025}. Several attempts have been developed to overcome the trainability issues~\cite{sweke2020stochastic,larocca2022diagnosing,li2024ensemble,chen2025quantum,Cerezo2021Cost,liu2024mitigating,sack2022avoiding,Pati2021entanglement,sannia2024engineered,yao2025avoiding}. On the other hand, using less expressive ansatz, e.g. through using shallow circuits, makes the training and experimental implementation easy though comes with the risk that the target state is not in the reachable set of the ansatz state. Indeed, designing an ansatz with a suitable reachable set that contains the target state, ideally without significantly increasing the circuit depth or the number of parameters, remains a key objective in advancing VQAs.

So far, most efforts to improve VQAs have focused on circuit design for $U(\bm\theta)$. Typical strategies include hardware-efficient ansatz (HEA)~\cite{Kandala2017}, Hamiltonian variational ansatz (HVA)~\cite{farhi2014quantumapproximateoptimizationalgorithm}, unitary coupled clusters~\cite{Bartlett1989133,Romero_2019,Lee2019kUpCCGSD,PRXQuantum.2.030301,Grimsley2020TrotterizedUCCSD,PhysRevA.98.022322,PhysRevResearch.2.033421, 10.1063/1.5141835}, symmetry preserving circuits~\cite{PhysRevApplied.16.034003,Gard2020SymmetryPreserving,PhysRevA.101.052340,PhysRevResearch.3.013039,Zheng_2023,Lyu2023symmetryenhanced,PhysRevResearch.6.013015,meyer2023exploiting}, tensor-network-inspired ansatz~\cite{Huggins_2019,PhysRevResearch.6.023009,Miao2025convergencequantum}, adaptive circuits~\cite{Grimsley2019,PhysRevResearch.6.013254,Gomes2021Adaptive,D0SC06627C,Ryabinkin2018QCC,Tang2021,Zhang_2021}, quantum circuits optimized by evolutionary algorithms~\cite{934383,PhysRevA.105.052414,evolutionary2022Ding}, and machine learning enhanced quantum circuits~\cite{NEURIPS2021_97244127,Zhang_2022,Du2022QCAS,shen2023prepareansatzvqediffusion}. By contrast, the role of the input state $|\Psi_0\rangle$ has received comparatively little attention, with a few exceptions, such as Ref.~\cite{Schuld2021effect}, 
despite its direct influence on the reachable set.

In this work, we address this overlooked aspect by focusing on the input-state design in enhancing the performance of VQAs. Specifically, we introduce a measurement-based encoder that prepares the input state as a superposition of candidate states, ensuring that optimization starts from a more favorable initialization. This strategy reshapes the reachable set without altering the original circuit structure, achieving systematic performance gains under the same gate budget. Because the parameterized circuit $U(\bm{\theta})$ remains unchanged, this approach applies broadly to a wide range of ansatz (e.g., HEA, HVA). We provide a theoretical guarantee that such input-state design improves the ground-state preparation of a given Hamiltonian, and validate this claim numerically on typical quantum many-body models. For the one-dimensional transverse-field Ising model, our method achieves fidelity $0.99$ with only $8$ layers, compared with $12$ layers for a conventional HEA, representing a one-third reduction in layers at the same fidelity. Similarly, for two-dimensional Ising and cluster-Ising models, it reaches the same fidelity $0.99$ with fewer circuit layers than typical HVAs. We further demonstrate consistent improvements in strongly correlated systems, including the Fermi-Hubbard model. These results demonstrate input-state design as a powerful and broadly applicable tool for building VQAs that are both expressive and trainable.

This paper is organized as follows. In Sec.~\ref{pre_knowledge}, we provide the background on VQEs, covering the two widely-used ansatz: the HEA and the HVA. In Sec.~\ref{theory}, we introduce our input-state design method, including algorithm design and a theoretical analysis that prescribes which candidate states to combine. As applications, in Sec.~\ref{state_preparation} we apply our method to typical problems in quantum many-body physics and compare our results with those achieved by HEA and HVA. For completeness, we also include an example on how to implement the encoder layer of our approach in the appendix.

\begin{figure*}[t]
\centering
\includegraphics[width=0.9\linewidth]{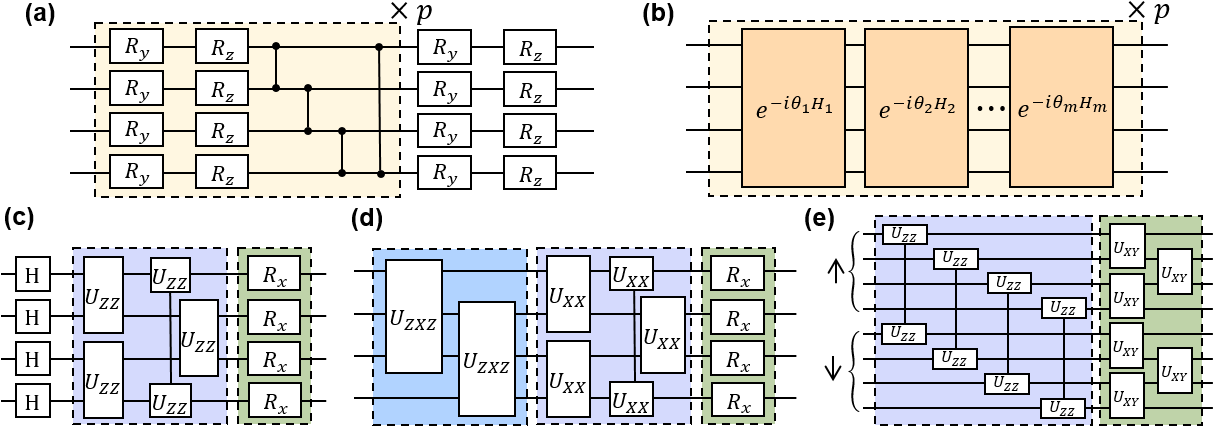}
\caption{%
(a) \textbf{Hardware-efficient ansatz (HEA).} Each layer consists of alternating single-qubit rotations $R_y$ and $R_z$ followed by a chain of $\mathrm{CZ}$ gates. The dashed box indicates one circuit layer, which is repeated $p$ times.  %
(b) \textbf{General Hamiltonian variational ansatz (HVA).} Each layer contains a product of unitaries $\prod_{k=1}^{q} e^{-i \theta_k H_k}$, where $\{H_k\}$ are problem-specific Hamiltonian terms. 
(c-e) \textbf{Examples of HVA design for three different models.} 
(c) \textit{For the transverse-field Ising model.} An initial layer of Hadamard gates $H$ prepares $\lvert+\rangle^{\otimes n}$. $U_{ZZ}(\theta) = e^{-i(\theta/2)\,\sigma_i^z\sigma_j^z}$ represents the two-qubit gate for ZZ interaction, while $R_x(\theta)=e^{-i\theta\,\sigma_i^x}$ represents the single-qubit $X$-rotation. 
(d) \textit{For the cluster-Ising model.} $U_{ZXZ}(\theta)=e^{-i(\theta/2)\,\sigma_i^z\sigma_j^x\sigma_k^z}$ is a three-qubit gate, and $U_{XX}(\theta)=e^{-i(\theta/2)\,\sigma_i^x\sigma_j^x}$ is a two-qubit gate. 
(e) \textit{For the Fermi-Hubbard model.} The upper (lower) register encodes spin-$\uparrow$ (spin-$\downarrow$). On-site interactions between the two spins at site $i$ are implemented as $U_{ZZ}(\theta)$. Hopping terms on odd and even bonds are realized by $U_{XY}(\theta)=e^{-i(\theta/2)\,(\sigma_i^x\sigma_{i+1}^x+\sigma_i^y\sigma_{i+1}^y)}$.
}
\label{fig:ansatz}
\end{figure*}

\section{Variational Quantum Eigensolver}\label{pre_knowledge}
The variational quantum eigensolver (VQE) \cite{McClean_2016,Peruzzo2014,PhysRevA.92.042303} is a special case of VQAs in which the observable $H$ is the Hamiltonian of a physical system. Therefore, the target state becomes the ground state of the Hamiltonian and the optimized cost function $E(\bm{\theta}_{\mathrm{opt}})$ will estimate the ground energy of the system. The resources required for the VQE, and in general for any VQA, can be divided into two categories: (i) quantum resources; and (ii) classical resources. The former can be quantified by the number of layers or the number of two-qubit gates in the circuit. The latter quantifies the resources that one needs for the classical optimization. This is proportional to the number of experiments that one has to perform to accomplish a VQE algorithm. As shown in Ref.~\cite{Lyu2023symmetryenhanced}, the classical resources can be quantified as $C_{\rm R}=N_{\rm I}\times N_{\rm para} $, where $N_{\rm I}$ is the average number of iterations that one has to use to converge the optimization and $N_{\rm para}$ is the number of parameters in the circuit which are optimized. In the following we will discuss some of the typical quantum circuits that are used in VQAs.

\subsection{Hardware-Efficient Ansatz}

The HEA~\cite{Kandala2017} consists of alternating layers of parameterized single-qubit rotations and entangling gates, arranged to match the native connectivity of the quantum hardware. As shown in Fig.~\ref{fig:ansatz}(a), it is defined as $U_{\mathrm{HEA}}(\bm{\theta}) {=} \prod_{l=1}^{p} \left[ U_{\mathrm{entangle}}^{(l)} \cdot \bigotimes_{i=1}^{n} R(\bm{\theta}_{l,i}) \right],$ where $p$ is the number of layers, $n$ is the number of qubits, and $U_{\mathrm{entangle}}^{(l)}$ denotes the $l$-th entangling layer, typically implemented using chains of controlled-$Z$ gates. The single-qubit rotation applied to qubit $i$ in layer $l$ is $R(\bm{\theta}_{l,i}) {=} R_y(\theta^{(y)}_{l,i})\, R_z(\theta^{(z)}_{l,i})$, with $R_y(\theta) {=} e^{-i\theta\sigma^y / 2}$ and $R_z(\theta) {=} e^{-i\theta\sigma^z / 2}$ denoting rotations about the $y$- and $z$-axes, respectively, and $\sigma^y$, $\sigma^z$ being Pauli operators. The main strength of HEA lies in its adaptability to different hardware platforms. However, its lack of problem-specific structure often leads to inefficient parameterization, making it susceptible to optimization challenges such as barren plateaus in deep circuits.

\subsection{Hamiltonian Variational Ansatz}

The HVA is inspired by the quantum approximate optimization algorithm (QAOA) and adiabatic quantum computation~\cite{PhysRevA.92.042303,farhi2014quantumapproximateoptimizationalgorithm,farhi2000quantumcomputationadiabaticevolution}. Unlike HEA, it exploits the structure of the problem Hamiltonian to guide ansatz circuit design and improve optimization efficiency. This approach is particularly suitable for Hamiltonians that can be expressed as a sum of local terms, $H {=} \sum_{i=1}^{q} H_i$, where each term $H_i$ is a Hermitian operator representing a local interaction, and all terms either commute or can be efficiently handled via Trotter-Suzuki decomposition. As shown in Fig.~\ref{fig:ansatz}(b), the HVA is constructed as a sequence of alternating unitary evolutions generated by these terms: $U_{\mathrm{HVA}}(\bm{\theta}) {=} \prod_{l=1}^p \prod_{i=1}^q e^{-i \theta_{l,i} H_i}$, where $p$ is the number of layers (or Trotter steps) and $\theta_{l,i}$ is the variational parameter associated with $H_i$ in layer $l$.

By incorporating the problem Hamiltonian into the ansatz, HVA naturally preserves relevant symmetries and reduces the effective optimization space, often leading to faster convergence and higher accuracy. Compared with HEA, it generally offers superior performance on structured problems, but may require deeper circuits depending on the Hamiltonian~\cite{10.1063/5.0186205}.

\section{Input-State Design}\label{theory}

For a given ansatz circuit $U(\boldsymbol{\theta})$ with a fixed number of layers, applying it to the standard input state $\ket{\Psi_0} = \ket{0}^{\otimes n}$ explores a restricted set of states localized near $\ket{\Psi_0}$ in the Hilbert space. This set is referred to as the reachable set of $\ket{\Psi_0}$ under $U(\bm{\theta})$ (red-shaded region in Fig.~\ref{fig:init}). If the target state $\ket{\Psi_{\mathrm{tar}}}$ lies outside this reachable set, then no optimization can achieve it, regardless of the choice of parameters. To overcome this limitation, we introduce a systematic \emph{input-state design} strategy that constructs a more suitable input state $\ket{\Psi_0}$, thereby modifying the reachable set to include $\ket{\Psi_{\mathrm{tar}}}$. Our approach employs an additional parameterized circuit $V(\bm{\gamma})$, termed the \emph{encoder}, whose role is to prepare a superposition input state $\ket{\Psi_0(\bm\gamma)}=V(\bm{\gamma})\ket{0}^{\otimes n}=\sum_{j=1}^m \alpha_j\,|\psi_j\rangle$, from which the target state is more accessible. Here, $\{|\psi_j\rangle\}$ is a pre-selected subset of computational basis states. The encoder parameters in $V(\bm{\gamma})$ and the ansatz parameters in $U(\bm{\theta})$ are then jointly optimized. A key advantage of this approach is that the encoder is a low-depth circuit, so it only modestly increases the overall circuit depth (see Appendix for details). As illustrated in Fig.~\ref{fig:init}, this design enables the same ansatz $U(\bm{\theta})$ to access an alternative region of the Hilbert space (green-shaded area) that contains the target state, thereby enhancing the expressive power of the original ansatz. Such modification of the reachable set has the potential to improve the ability of the original VQA in approximating the true ground state with higher fidelity. 

To formalize this intuition, we next present a theorem establishing how the achievable ground-state fidelity of a linear combination of orthogonal input states relates to the fidelity achievable from each individual input. This result provides a rigorous foundation for why carefully engineered input states can systematically modify the reachable set and improve performance.

\begin{theorem}\label{theorem1}
Let $\mathcal{A}_m{=}\{\,|\psi_j\rangle\,\}_{j=1}^{m}$ denote the $m$ selected mutually orthogonal states, and $\ket{\Psi}{=}\sum_{j=1}^m \alpha_j \ket{\psi_j}$. Then the fidelity with respect to the target state (which is the true ground state of the Hamiltonian), i.e., $F{=}|\bra{\Psi} \Psi_\mathrm{tar}\rangle|^2$, can be maximized over $\bm{\alpha}$ and satisfies:
\begin{align}
\max_{\bm{\alpha}} F = \sum_{j =1}^m F_j,
\end{align}
where $ F_j {\equiv} |\langle \psi_j | \Psi_{\mathrm tar} \rangle|^2 $, and the optimal parameters $ \bm{\alpha} $ and $ \bm{\beta} $ are collinear, with $ \bm{\beta} {\equiv} (\langle \psi_1 | \Psi_{\mathrm tar} \rangle, \langle \psi_2 | \Psi_{\mathrm tar} \rangle, \dots, \langle \psi_m | \Psi_{\mathrm tar} \rangle)^T $.
\end{theorem}

\begin{proof}
Let $\ket{\Psi} = \sum_{j=1}^m \alpha_j \ket{\psi_j}$ be the trial state. The fidelity is:
\begin{align}
F \equiv |\langle \Psi | \Psi_{\mathrm tar} \rangle|^2
= \left| \sum_{j=1}^m \alpha_j^{\!*} \langle \psi_j | \Psi_{\mathrm tar} \rangle \right|^2
= |\bm{\alpha}^\dagger \bm{\beta}|^2.
\end{align}
By the Cauchy--Schwarz inequality,
\begin{align}
|\bm{\alpha}^\dagger \bm{\beta}|^2 \leq \|\bm{\alpha}\|^2 \,\|\bm{\beta}\|^2,
\end{align}
with equality if and only if $\bm{\alpha}$ and $\bm{\beta}$ are collinear. Given the normalization constraint $\|\bm{\alpha}\|^2 = 1$, the fidelity achieves its maximum, and the maximum value is
\begin{align}
F = \|\bm{\beta}\|^2 = \sum_{j \in \mathcal{A}_m} |\langle \psi_j | \Psi_{\mathrm tar} \rangle|^2 = \sum_{j \in \mathcal{A}_m} F_j.
\end{align}
\end{proof}

\begin{figure}[t]
\centering
\includegraphics[width=0.7\linewidth]{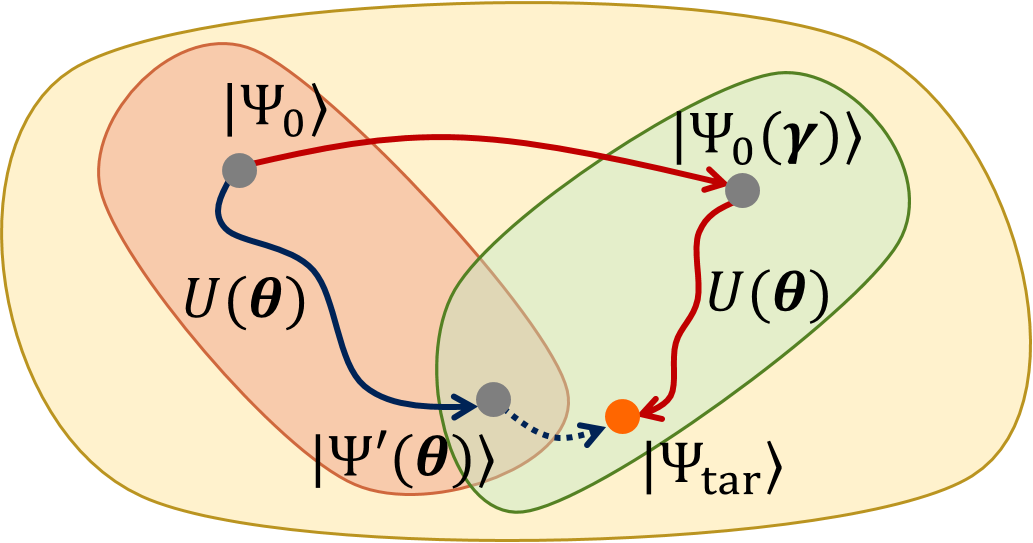}
\caption{Reachable sets modified through input-state design. For a fixed unitary $U(\bm{\theta})$, a simple input state $\ket{\Psi_0}$ induces a reachable set (red-shaded) that excludes the target $\ket{\Psi_{\mathrm{tar}}}$, causing optimization to converge to a suboptimal state $\ket{\Psi'(\bm{\theta})}$ (blue path). By contrast, a designed input state $\ket{\Psi_0(\bm{\gamma})}$, prepared by the encoder $V(\bm{\gamma})$, produces a different reachable set (green-shaded) that contains $\ket{\Psi_{\mathrm{tar}}}$, enabling the same $U(\bm{\theta})$ to reach the target (red path).}
  \label{fig:init}
\end{figure}

Theorem~\ref{theorem1} reveals that when combining multiple orthogonal quantum states $\{\ket{\psi_j}\}$, each state's contribution to the fidelity is independent, and the maximum achievable fidelity equals the sum of the individual fidelity $F_j$. This is because unitary operation preserves the orthogonality of the quantum states which contribute in the linear superposition. This insight underpins the selection strategy adopted later: prioritize candidates with larger $F_j$ so that the encoder-prepared superposition attains maximal overlap with the target within the chosen subspace.

With Theorem~\ref{theorem1}, the next task is to implement an encoder $V(\bm\gamma)$ that prepares the optimal input-state $\ket{\Psi_0(\bm\gamma)}{=}\sum\alpha_j \ket{\psi_j}$ from a simple initial state. Since the result assumes that the candidate states are mutually orthogonal, the set $\{\ket{\psi_j}\}$ is instantiated with the computational basis states $\{\ket{j}\}$, which form an orthonormal set and are hardware-friendly to prepare. We now detail the practical implementation. While the framework is motivated by maximizing fidelity $F$, on hardware fidelity is not directly accessible. We therefore adopt the standard VQE objective and minimize the variational average energy, and our algorithm comprises the following six steps: 

1. The standard ansatz circuit $U(\bm{\theta})$ is trained from the simple initial state $|0\rangle$ by minimizing $E(\bm{\theta}) {=} \langle 0 | U^{\dagger}(\bm{\theta}) H U(\bm{\theta}) | 0 \rangle$. We denote the resulting optimized parameters by $\tilde{\bm{\theta}}_\mathrm{opt}$. The optimization ends when the gradient norm is less than $10^{-3}$. 

2. A total of $M$ computational basis states $\{ |j^{(k)}\rangle \}_{k=1}^M$ are randomly sampled from the full Hilbert space. This heuristic sampling strategy enables efficient exploration of a representative subspace while reducing computational overhead. The state $|0\rangle$ is always included in the sampled set. 

3. For each sampled basis state $|j^{(k)}\rangle$, we estimate the energy expectation value $E_{j^{(k)}} {=} \langle j^{(k)} | U^{\dagger}(\tilde{\bm{\theta}}_\mathrm{opt}) H U(\tilde{\bm{\theta}}_\mathrm{opt}) | j^{(k)} \rangle$. $E_{j^{(k)}}$ is obtained via quantum measurements, where the number of measurements required scales as $N_m = 1/\epsilon^2$ for a target estimation error $\epsilon$, as discussed in~\cite{Wu2025quantum}. 

4. Using the energy estimates $E_{j^{(k)}}$ from step~3, we first build a candidate pool by retaining basis states with $E_{j^{(k)}}< T_e$, where $T_e$ defines the low-energy range for selection. From this pool we randomly select $m$ states, and we always include the reference state $|0\rangle$ to form $\mathcal{A}_m$. Here, we choose $m$ to scale linearly with the system size (for example, $m=6$ for 12 qubits). This design is motivated by Theorem~\ref{theorem1}, which ensures that an appropriate combination of the selected states achieves a large overlap with the target state. Importantly, we do not simply pick the $m$ lowest-energy states, because the smallest expectation values do not guarantee a large overlap with the true ground state.

5. Based on the set of computational basis states $\mathcal{A}_m$ selected in the previous step, the encoder $V(\bm{\gamma})$ is constructed such that $V(\bm{\gamma})|0\rangle {=} \sum_{k=1}^m \alpha_k |j^{(k)}\rangle.$ Following~\cite{9586240}, such an encoder can be implemented efficiently. For the same circuit $U(\bm{\theta})$, the basis set $\mathcal{A}_m$ is fixed, so the encoder $V(\bm{\gamma})$ only needs to be designed once and can then be reused throughout the optimization. A concrete construction is provided in Appendix.

6. Once the encoder layer $V(\bm{\gamma})$ is constructed, the input state of $U(\bm{\theta})$ becomes $\ket{\Psi_0(\bm\gamma)}{=}V(\bm{\gamma})|0\rangle {=} \sum_{k=1}^m \alpha_k |j^{(k)}\rangle$. Subsequently, applying the circuit $U(\bm{\theta})$ yields $|\Psi(\bm{\theta,\bm{\gamma}})\rangle {=} U(\bm{\theta})V(\bm{\gamma})|0\rangle$, as illustrated in Fig.~\ref{fig:circuit_scheme}. We then jointly optimize $\bm{\theta}$ and $\bm{\gamma}$ to minimize $E(\bm{\theta},\bm{\gamma}) {=} \langle 0|V^\dagger(\bm{\gamma})U^\dagger(\bm{\theta})HU(\bm{\theta})V(\bm{\gamma})|0\rangle$. Here, the optimization parameters are initialized with $\tilde{\bm{\theta}}_\mathrm{opt}$, obtained from the pre-trained ansatz, and random initial encoder parameters $\bm{\gamma}$. To control classical overhead, we restrict this optimization to a modest number of iterations $T$ (typically $T=200$), though this can be increased for more difficult problems. The full procedure is summarized in Algorithm~\ref{alg:enhanced_vqe}.

\begin{figure*}[t]
\centering
\includegraphics[width=\linewidth]{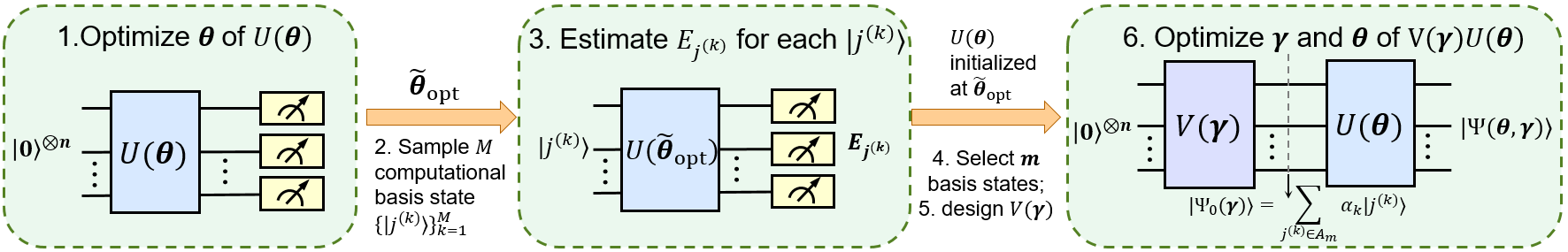}
\caption{%
\textbf{Workflow of the Input-state Design Algorithm.} 
(1) The circuit $U(\bm{\theta})$ is first pre-trained on the initial state $\ket{0}^{\otimes n}$ to obtain the optimized parameters $\tilde{\bm{\theta}}_{\mathrm{opt}}$. 
(2) With $U(\tilde{\bm{\theta}}_{\mathrm{opt}})$ held fixed, $M$ computational basis states $\{\ket{j}\}$ are randomly sampled. 
(3) For each $\ket{j}$, $E_j$ is evaluated under the pre-trained circuit. 
(4) According to a selection rule, $m$ promising candidates (including $\ket{0}^{\otimes n}$) out of the $M$ basis states are chosen to form the set $\mathcal{A}_m$. 
(5) An encoder $V(\bm{\gamma})$ is then constructed to prepare a superposition of these $m$ states. 
(6) The encoder $V(\bm{\gamma})$ and the original ansatz $U(\bm{\theta})$ are jointly optimized to minimize the energy, with the ansatz parameters initialized at $\bm{\theta}=\tilde{\bm{\theta}}_{\mathrm{opt}}$.
}
\label{fig:circuit_scheme}
\end{figure*}

To illustrate the practical significance of this result, we consider a $U(\bm\theta)$ that, after optimization, plateaus at a fidelity around $0.95$, beyond which increasing circuit depth or training iterations yields little improvement due to ansatz limitations or barren plateaus. However, according to Theorem~\ref{theorem1}, by constructing an appropriate combination of multiple basis states, it is possible to surpass this fidelity barrier and achieve significantly higher overlaps with the target state. In particular, in the limiting case where $m = 2^n$ and $\{\ket{\psi_j}\}$ span the entire Hilbert space, any target state $\ket{\Psi_{\mathrm tar}}$ can be precisely reconstructed, yielding perfect fidelity $F = 1$. Of course, constructing such a complete encoder requires exponential resources. Our method circumvents this overhead by selecting only a small subset of basis states with the highest $F_j$, thereby achieving a scalable and resource-efficient enhancement of fidelity. As we demonstrate in the following simulations, this theoretical insight is well supported by empirical results, confirming the effectiveness of our approach.

\begin{algorithm}[H]
\caption{Input-state Design Algorithm}
\label{alg:enhanced_vqe}
\raggedright
\begin{algorithmic}
\Require Hamiltonian $H$, initial state $\ket{0}$, initial unitary $U(\bm{\theta})$, number of sampled basis states $M$, number of selected basis states $m$.
\Ensure Optimized parameters $\bm{\theta}$ and $\bm{\gamma}$ minimizing the energy $E$.

\State \text{1:} Optimize $\bm{\theta}$ of $U(\bm{\theta})$ to obtain $\tilde{\bm{\theta}}_\mathrm{opt}$:
\State \quad $\tilde{\bm{\theta}}_\mathrm{opt}=\arg\min_{\bm{\theta}} 
       \langle 0|U^{\dagger}(\bm{\theta})HU(\bm{\theta})|0\rangle$.

\State \text{2:} Uniformly sample $M$ computational basis states 
\State \quad $\{\ket{j^{(k)}}\},k=1,2,\dots, M$.

\State \text{3:} For each sampled $\ket{j^{(k)}}$, estimate 
\State \quad $E_{j^{(k)}}=\bra{j^{(k)}}U^{\dagger}(\tilde{\bm{\theta}}_\mathrm{opt})HU(\tilde{\bm{\theta}}_\mathrm{opt})
        \ket{j^{(k)}}$.

\State \text{4:} Select $m$ states among those with $E_{j^{(k)}}<T_e$ to form 
\State \quad $\mathcal{A}_m=\{\ket{0}\}\cup\{\ket{j^{(k)}}\}_{k=1}^{m-1}$.

\State \text{5:} For given basis set $\mathcal{A}_m$, we can construct encoder $V(\bm{\gamma})$ 
\State \quad such that $V(\bm{\gamma})\ket{0} = \sum_{k=1}^m \alpha_k \ket{j^{(k)}}.$

\Statex \text{6:} \textbf{while} Iteration $t \leq T$ \textbf{do}
\State \quad Optimize $(\bm{\theta}, \bm{\gamma})$ to obtain
\State \quad $(\bm{\theta}_{\mathrm{opt}},\bm{\gamma}_{\mathrm{opt}} )= \arg\min_{\bm{\theta},\bm{\gamma}}
   \bra{0}V^{\dagger}(\bm{\gamma})U^{\dagger}(\bm{\theta})
     H U(\bm{\theta})V(\bm{\gamma})\ket{0}$,
\State \quad $U(\bm{\theta})$ initialized at $\tilde{\bm{\theta}}_\mathrm{opt}$.

\State \textbf{end while}

\State \Return Optimized parameters $\bm{\theta}_{\mathrm{opt}}$ and $\bm{\gamma}_{\mathrm{opt}}$.
\end{algorithmic}
\end{algorithm}


\section{Input-State Design in Quantum Many-Body Problems}\label{state_preparation}

As applications, we present numerical simulations of ground state preparation for the Ising model, the cluster-Ising model, and the Fermi-Hubbard model to demonstrate the efficiency of input-state design. Rather than increasing circuit depth or changing the ansatz, we design the input state $\ket{\Psi_0(\bm{\gamma})}$ so that the reachable set of the same parameterized circuit $U(\bm{\theta})$ is shifted to include the target ground state. This reachable set reconfiguration improves both expressivity and trainability: under fixed depth and fixed circuit structure (HEA or HVA), the designed inputs yield systematically higher fidelity and more accurate ground energy estimates than standard choices. These results indicate that controlling $\ket{\Psi_0(\bm{\gamma})}$ is a practical lever for overcoming expressivity bottlenecks without deepening circuits, a point we substantiate across the three models above.

\subsection{Transverse-Field Ising Model}

\begin{figure*}[tp]
  \centering
  \includegraphics[width=\linewidth]{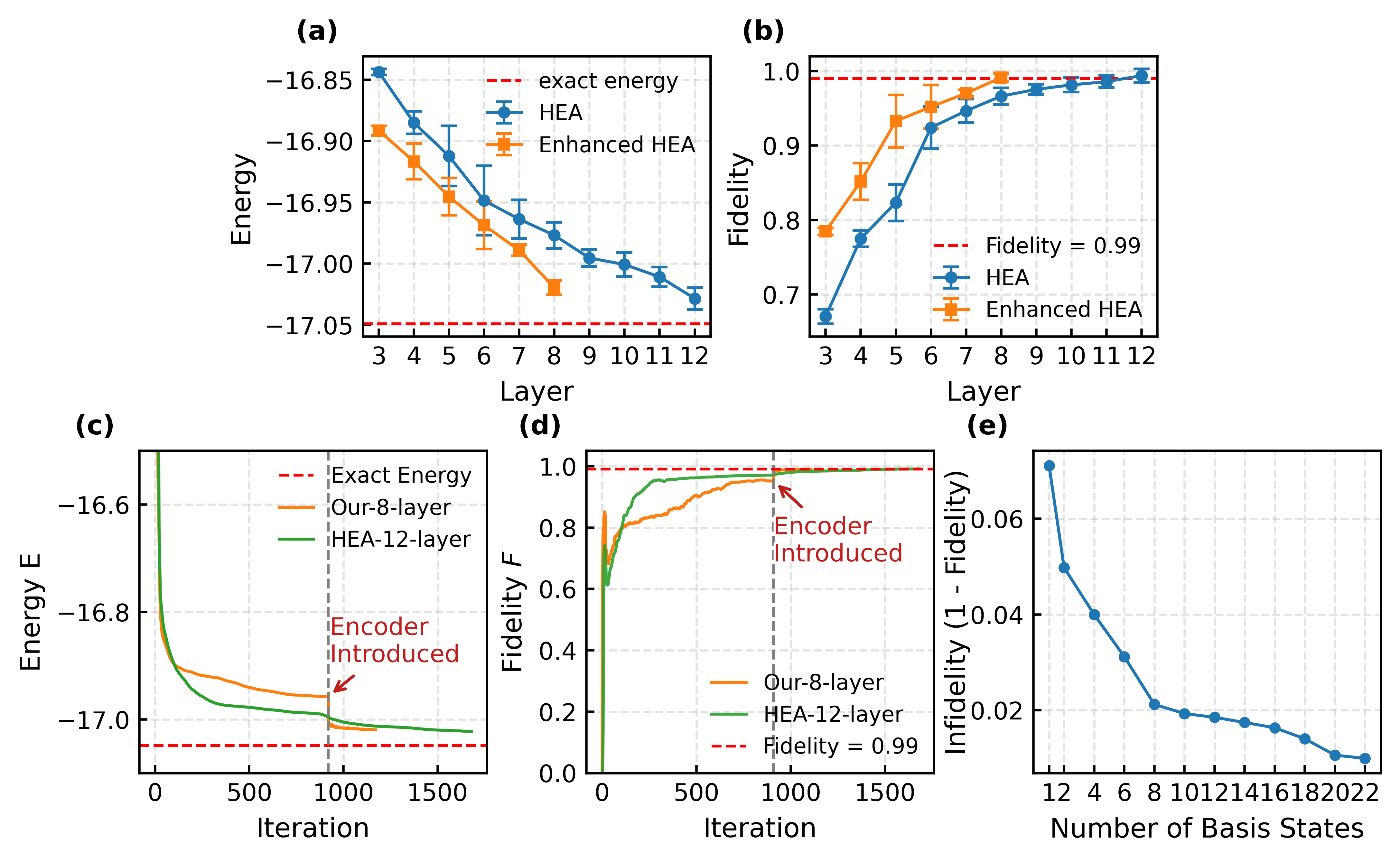}
\caption{
Simulation results for the one-dimensional transverse-field Ising model using conventional HEA and our input-state design strategy (enhanced HEA). (a)-(b) Performance comparison of variational circuits with different layer depths, showing (a) ground energy and (b) fidelity. The enhanced circuit (orange) reaches an average fidelity of $0.99$ with only 8 layers, while the conventional HEA circuit (blue) requires $12$ layers to achieve the same accuracy. (c)-(d) Training trajectories for the 8-layer enhanced HEA are compared with a 12-layer conventional HEA baseline. The gray dashed line marks the iteration at which the encoder is introduced. After this point, the energy quickly approaches the exact ground energy and the fidelity continues rising. (e) Infidelity (i.e., $1-F$) as a function of the number of computational-basis states selected for the encoder, based on a 5-layer HEA circuit. As $m$ increases, the infidelity decreases approximately exponentially, showing that as few as $m = 6$ carefully chosen states are sufficient for high-fidelity state preparation, enabling expressive yet resource-efficient initialization.}
\label{fig:ising}
\end{figure*}

We consider the transverse-field Ising model (TFIM) with periodic boundary conditions,
\begin{equation}
H = -J  \sum_{\langle i,j \rangle} \sigma_i^z \sigma_j^z - h \sum_{i} \sigma_i^x ,
\end{equation}
where $\langle i,j \rangle$ denotes nearest neighbors and $\sigma_i^{x,z}$ are Pauli operators. The parameter $J$ sets the spin-spin coupling: $J>0$ corresponds to ferromagnetic and $J<0$ to antiferromagnetic interactions. The transverse field strength $h$ controls the external field along $x$. In 1D the model is exactly solvable via, e.g., the Jordan--Wigner transformation~\cite{LIEB1961407}, whereas in 2D increased connectivity renders the problem substantially more challenging, with progress primarily through numerical methods~\cite{slattery2021quantumcircuitstwodimensionalisometric,PhysRevB.110.155128}.

\paragraph{1D case.}  
As a first example, we consider the antiferromagnetic regime with $J=-1$ and $h=-1.2$. The baseline circuit $U(\bm{\theta})$ uses $p$ layers HEA circuit, see Fig.~\ref{fig:ansatz}(a). The original input state is taken to be $\ket{0}^{\otimes n}$. Our input-state design method keeps the total depth $p$ fixed and replaces one HEA layer with an input-state encoder $V(\bm{\gamma})$. The encoder is built from $m=6$ computational-basis states selected from $M=2000$ candidates and its gate cost comparable to a single HEA layer. For a 12-qubit system, one HEA layer requires about $60$ basic gates; with $m=6$, the encoder uses $\sim 40$-$60$ gates depending on the chosen basis states. Figs.~\ref{fig:ising}(a)-(b) report the ground energy and fidelity versus depth, with each point averaged over $100$ random initializations. At matched depth and circuit structure, input-state design method yields uniformly higher fidelity than the HEA baseline. In particular, our method reaches $0.99$ fidelity at $8$ layers (with $112$ CNOT gates), whereas the standard HEA requires $12$ layers (with $144$ CNOT gates), a $\sim 33\%$ reduction in depth.

\paragraph{Comparing performance with the conventional HEA.}  

We next compare the performance of the two methods at fixed circuit depths. Our method is tested at total depths of $8$ (one encoder layer plus $7$ HEA layers). For reference, we include a standard HEA baseline with $12$ layers, the minimum needed to reach $0.99$ fidelity. Our method training proceeds in two stages: first, only the HEA parameters are optimized until the gradient norm drops below $10^{-3}$, signaling barren plateau onset. Then the encoder is introduced, and joint optimization of both $\bm{\theta}$ and $\bm{\gamma}$ runs for $200$ iterations. Figs.~\ref{fig:ising}(c)-(d) show this handoff (dashed line), after which the energy rapidly approaches the exact ground energy and the fidelity increases. With only $8$ layers, the enhanced circuit achieves a fidelity of $0.99$, whereas the conventional HEA requires $12$ layers.

We also compare the quantum and classical resources (Table~\ref{tab:resources}), using the metrics defined in Sec.~\ref{pre_knowledge}. The enhanced circuit reaches target fidelity with $8$ layers (containing $112$ CNOT gates) and $1100$ optimization steps, while HEA needs $12$ layers ($144$ CNOT gates) and $1500$ steps. We find that the input-state design method reduces both gate count and optimization effort. Thus, the encoder yields significant efficiency gains in both quantum and classical resources.


We further investigate how performance depends on $m$, the number of selected basis states. Fig.~\ref{fig:ising}(e) shows that larger $m$ improves final fidelity but with diminishing returns. Substantial gains appear up to $m \approx 6\text{-}8$, after which improvements saturate. In practice, we set $m$ by accounting for the gate cost of implementing the encoder $V(\bm\gamma)$: we choose $m$ so that the additional gates for $V(\bm\gamma)$ are comparable to, or below, the cost of one layer of the baseline ansatz. In this case, one HEA layer uses $60$ gates, so we take $m=6$, which requires only $40\text{-}60$ gates to realize $V(\bm\gamma)$. Hence, moderate $m$ captures most of the benefit while avoiding unnecessary overhead.


\begin{figure*}[t]
\centering
\includegraphics[width=\linewidth]{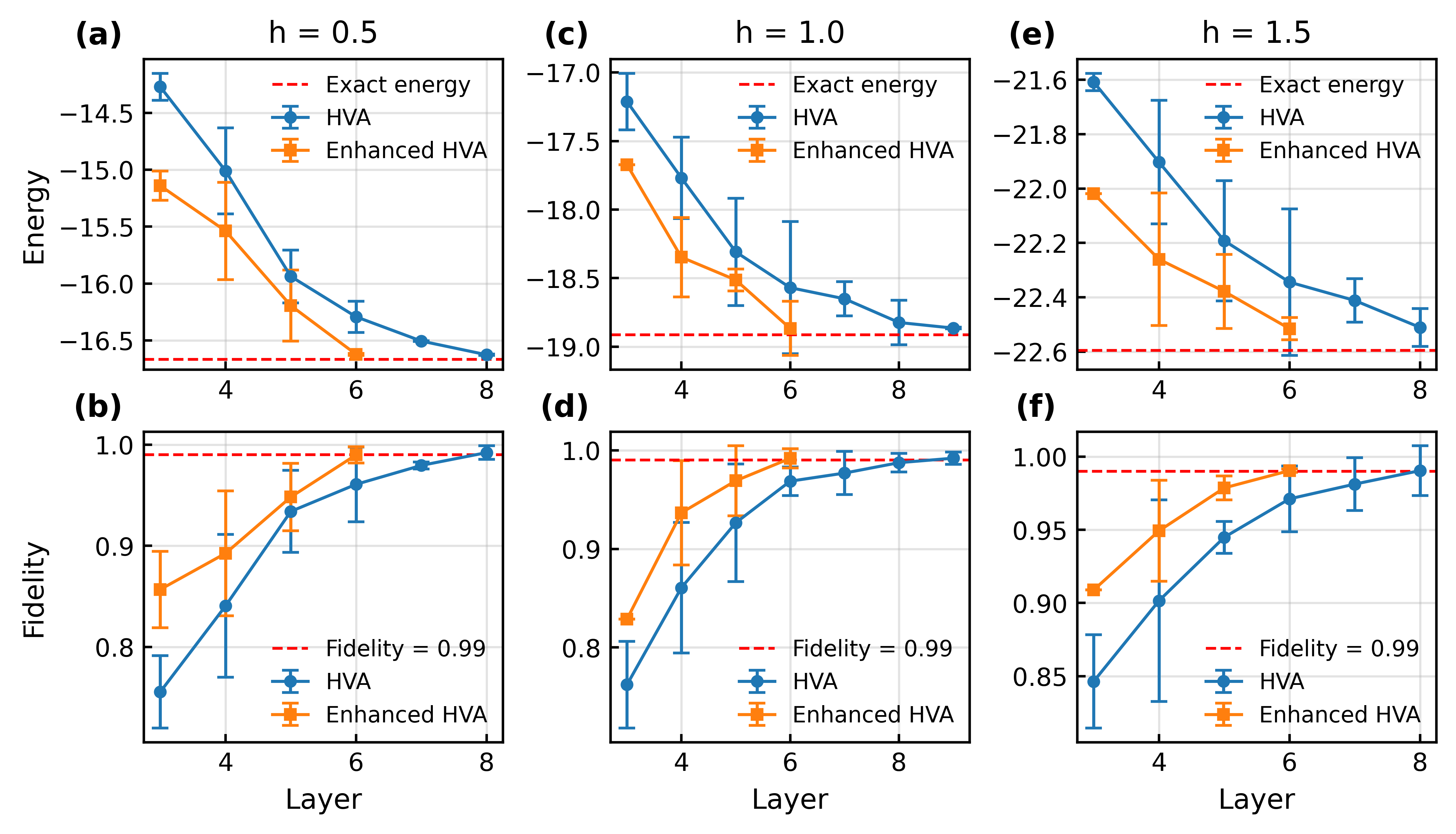}
\caption{
Simulation results for the 12-qubit 2D Ising model at $h = 0.5$, $1$, and $1.5$. For each field strength $h$, the upper panel shows the ground energy as a function of circuit depth $p$, and the lower panel reports the corresponding fidelity to the exact ground state. The blue curves correspond to the conventional HVA, and the orange curves correspond to the input-state design (enhanced HVA). Each marker represents the mean over $100$ random initializations. Across all three values of $h$, the input-state design consistently achieves lower variational energies and higher fidelities under the same depth, and it reaches the $0.99$ fidelity threshold with fewer layers than the baseline.
}
\label{Ising_HVA}
\end{figure*}

\paragraph{2D case.}  
As the second example, we apply our method to the 2D TFIM with periodic boundary conditions, using the HVA in Fig.~\ref{fig:ansatz}(c). The original input state is $\ket{+}^{\otimes n}$. Each HVA layer is $U_{\mathrm{HVA}}(\bm\theta) {=} e^{-i\theta_1 H_{ZZ}^{\mathrm{(odd)}}} \, e^{-i\theta_2 H_{ZZ}^{\mathrm{(even)}}} \, e^{-i\theta_3 H_X}$, with $H_{ZZ}^{\mathrm{(odd)/(even)}}{=} {-}\sum_{\langle i,j\rangle} \sigma_i^z\sigma_j^z$ and $H_X {=} {-}\sum_i \sigma_i^x$. Note that in 2D we traverse both horizontal and vertical couplings. 

We simulate three field strengths, $h \in \{0.5, 1.0, 1.5\}$, and enhance the HVA by adding an encoder constructed from $m=8$ basis states selected from $M=2000$ candidates. Each HVA layer uses $\sim 84$ gates (including $48$ two-qubit gates and $36$ single-qubit gates), so the encoder’s cost remains comparable to one HVA layer. For each configuration, $100$ random initializations are performed. As shown in Fig.~\ref{Ising_HVA}, the enhanced HVA consistently outperforms the conventional version, achieving lower variational energy and higher fidelity at lower depths across all values of $h$. 

Together, the 1D and 2D results demonstrate that our input-state design improves both efficiency and accuracy across different ansatz, establishing it as a broadly applicable strategy for overcoming expressivity bottlenecks.

\subsection{Cluster-Ising Model}

\begin{figure}[t]
\centering
\includegraphics[width=0.9\linewidth]{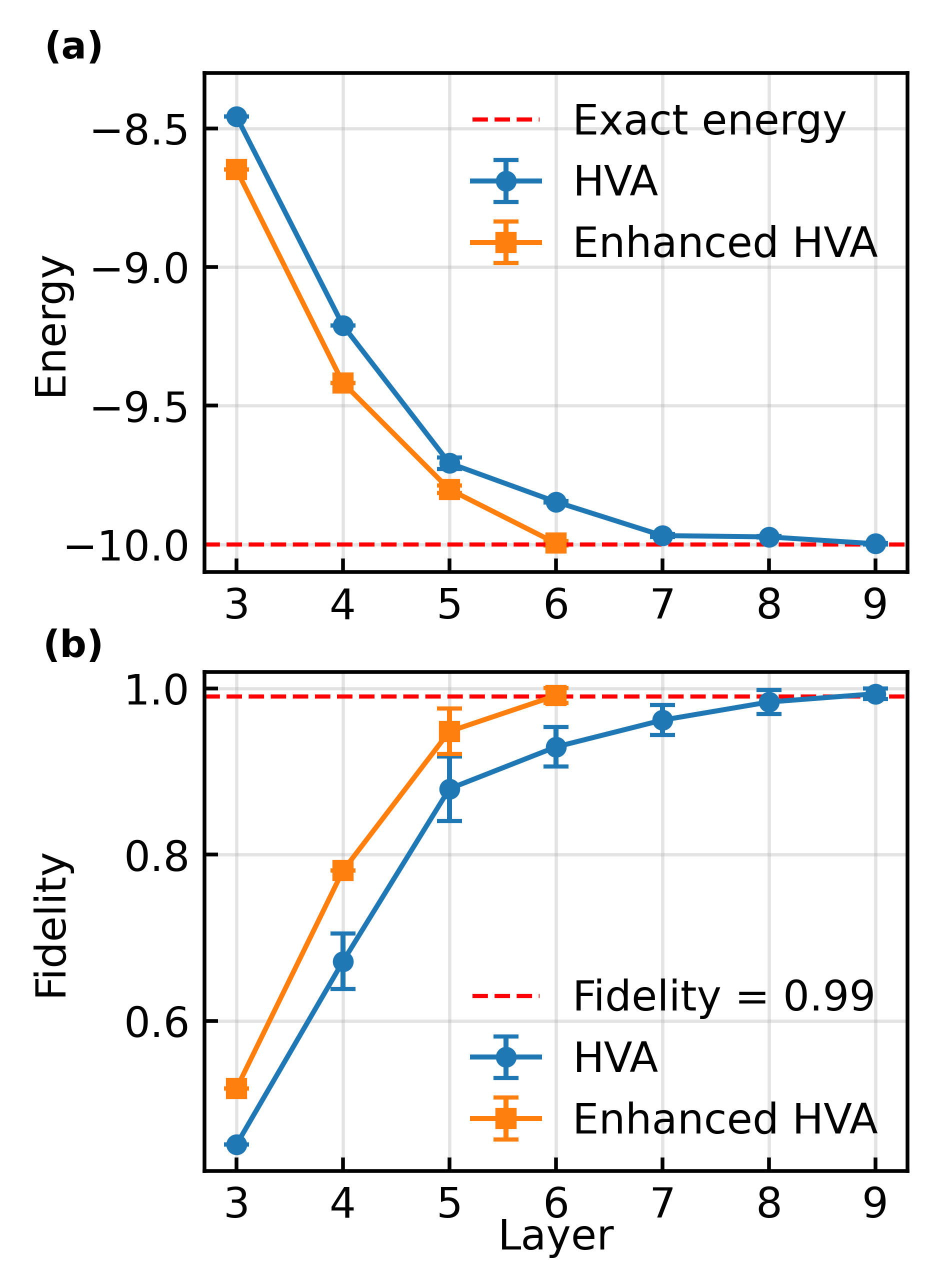}
\caption{
Simulation results for the 12-qubit cluster-Ising model at $h = 0.5$, $1.0$, and $1.5$. The top and bottom rows respectively show the ground energy and fidelity as functions of circuit depth $p$. The conventional HVA (blue) is compared with an enhanced HVA with the input-state design (orange). Each data point represents the average over $100$ random initializations. Across all values of $h$, the input-state design consistently achieves lower variational energy and higher fidelity at the same circuit depths.
}
\label{Cluster_Ising_HVA}
\end{figure}

In addition to the models discussed above, we also consider other quantum systems to further evaluate the generality of our method. Next, we consider the cluster-Ising model described by the Hamiltonian
\begin{equation}
\begin{aligned}
H &= -J \sum_{i=1}^{N-2} \sigma_{i}^{z}\,\sigma_{i+1}^{x}\,\sigma_{i+2}^{z} \\
  &\qquad - h_{1} \sum_{i=1}^{N-1} \sigma_{i}^{x}\,\sigma_{i+1}^{x}
           - h_{2} \sum_{i=1}^{N} \sigma_{i}^{x}.
\end{aligned}
\end{equation}
Here, $J$ sets the strength of the three-body $ZXZ$ interaction; $h_1$ controls nearest-neighbor $XX$ coupling; and $h_2$ is a uniform transverse field along $x$. For $h_1=h_2=0$, this reduces to the 1D cluster model with the cluster state as its unique ground state, a symmetry-protected topological phase protected by a $Z_2\times Z_2$ symmetry~\cite{PhysRevLett.109.050402}. Turning on $h_1$ or $h_2$ breaks the protecting $Z_2\times Z_2$ symmetry and drives the system toward a non-topological phase.

To approximate the ground state of this Hamiltonian for general $(J,h_1,h_2)$, we use HVA. The Hamiltonian is first decomposed into groups of mutually commuting terms, each group generating a parameterized unitary operator:  
$H {=} J H_{ZXZ} + h_1 \left( H_{XX}^{\mathrm{(odd)}} + H_{XX}^{\mathrm{(even)}} \right) + h_2 H_{X}$,  
where 
$H_{ZXZ} {=} -\sum_{i=1}^{N-2} \sigma_{i}^{z}\sigma_{i+1}^{x}\sigma_{i+2}^{z}$, 
$H_{XX}^{\mathrm{(odd)}} {=} -\sum_{\substack{i=1 }}^{N-1} \sigma_{i}^{x}\sigma_{i+1}^{x}$, 
$H_{XX}^{\mathrm{(even)}} {=} -\sum_{\substack{i=2 }}^{N-2} \sigma_{i}^{x}\sigma_{i+1}^{x}$, 
and $H_{X} {=} -\sum_{i=1}^{N} \sigma_{i}^{x}$. 
As shown in Fig.~\ref{fig:ansatz}(d), each HVA layer $l$ consists of sequentially applying 
$e^{-i\theta^{(l)}_{1} H_{ZXZ}}$, 
$e^{-i\theta^{(l)}_{2} H_{XX}^{\mathrm{(odd)}}}$, 
$e^{-i\theta^{(l)}_{3} H_{XX}^{\mathrm{(even)}}}$, 
and $e^{-i\theta^{(l)}_{4} H_{X}}$, 
with independent variational parameters $\theta^{(l)}_{1\text{-}4}$ for each layer. For the three-body term $e^{-i\theta \sigma_{i}^{z}\sigma_{i+1}^{x}\sigma_{i+2}^{z} } $ in $H_{ZXZ}$, we use a native two-qubit gate decomposition compatible with typical hardware connectivity: 
$e^{-i\theta \sigma_{i}^{z}\sigma_{i+1}^{x}\sigma_{i+2}^{z} } =\mathrm{CZ}_{(i,i+1)}\, \mathrm{CZ}_{(i+2,i+1)}\, e^{-i\theta X_{i+1}}\, \mathrm{CZ}_{(i,i+1)}\, \mathrm{CZ}_{(i+2,i+1)}$, 
where $\mathrm{CZ}_{(a,b)}$ denotes a controlled-$Z$ gate between qubits $a$ and $b$.

We next evaluate the performance of the HVA in approximating the ground state of $H$ for the case $h_1 = h_2 = 0.1$. In our simulations, the original input state is $\ket{\psi_0} {=} \ket{+}^{\otimes n}$. From a pool of $M=2000$ candidate basis states we select $m=12$ to construct the enhanced HVA, choosing $m$ so that the encoder’s gate count remains comparable to one HVA layer (including $62$ two-qubit gates and $103$ single-qubit gates). For each depth, we perform $100$ independent optimizations from random initial parameters. As shown in Fig.~\ref{Cluster_Ising_HVA}, the enhanced method reaches fidelity $0.99$ with $p=6$ layers (requiring $450$ two-qubit gates in total), whereas the conventional HVA requires $p=9$ (requiring $558$ two-qubit gates). Moreover, a comparison between classical resources shows that the input-state design method requires $C_{\rm R}=54550$ while the conventional HVA algorithm demands $C_{\rm R}=118800$. This clearly shows the superiority of the input-state design approach over conventional methods.

\subsection{Fermi-Hubbard Model}
\begin{figure}[t]
\centering
\includegraphics[width=0.7\linewidth]{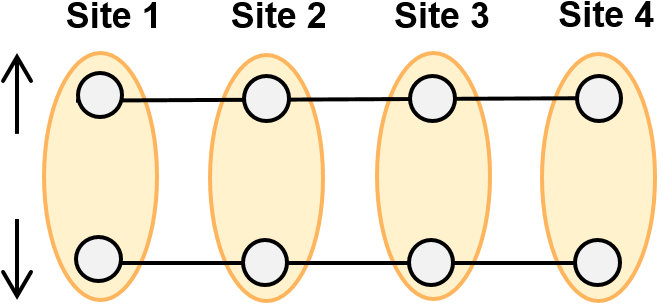}
\caption{
Fermi-Hubbard model for a $1\times4$ chain. Each site hosts spin-$\uparrow$ and spin-$\downarrow$ orbitals (gray). Shaded ovals indicate the on-site interaction $U$ between the two spins on the same site.}

\label{Hubbard_site}
\end{figure}

\begin{figure*}[t]
\centering
\includegraphics[width=\linewidth]{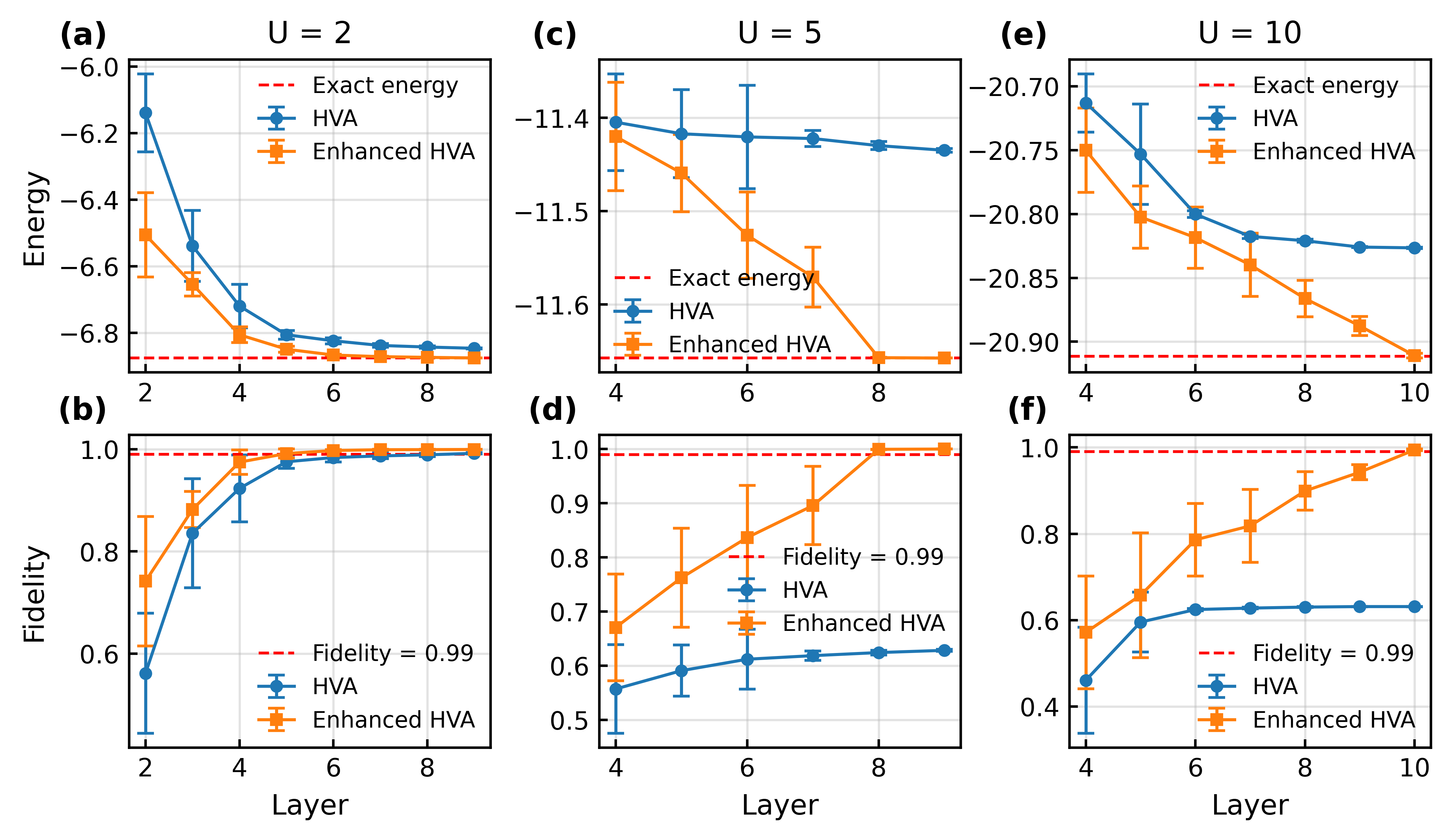}
\caption{
Simulation results for the 4-site (8-qubit) one-dimensional Hubbard model at half filling, for interaction strengths $U = 2$, $5$, and $10$. Each data point represents the average over 100 random initializations. The baseline (blue) uses HVA initialized using a single bitstring, while our input-state design (enhanced HVA, orange) employs a parameterized encoder to prepare a superposition input. At matched circuit depth, the encoder consistently achieves lower variational energies and higher fidelities, with the largest gaps at $U=5$ and $U=10$ where the baseline stagnates while our method continues to improve.
}
\label{Hubbard}
\end{figure*}

To further demonstrate the generality of our input state design strategy, we consider a third benchmark: the 1D Fermi-Hubbard model at half-filling, a paradigmatic example of strongly correlated quantum systems. This model serves as a stringent testbed for evaluating whether our method can generate expressive and physically meaningful input state that enhance ground-state fidelity and optimization performance.
The Hubbard model with open boundary conditions is described by the Hamiltonian
\begin{align}
H &= -t \sum_{i,\,s} \left( c^{\dagger}_{i s} c_{i+1,\,s} + \text{H.c.} \right) \nonumber \\
&\quad + U \sum_{i} \left( n_{i\uparrow} - \frac{1}{2} \right)\!\left( n_{i\downarrow} - \frac{1}{2} \right),
\label{eq:hubbard}
\end{align}
where $c_{i s}^\dagger$ ($c_{i s}$) creates (annihilates) a fermion with spin $s \in \{\uparrow,\downarrow\}$ at the site $i$, and $n_{i s} = c_{i s}^\dagger c_{i s}$ is the corresponding number operator. The hopping amplitude $t$ sets the energy scale and $U$ denotes the on-site Coulomb interaction.

To simulate this model on qubits, we apply the Jordan--Wigner (JW) transformation to map fermionic operators to qubit operators. Since each site hosts two spin species, a system with $N$ sites is encoded into $2N$ qubits. Specifically, we assign qubit $i$ to the spin-$\uparrow$ fermion at site $i$, and qubit $i{+}N$ to the spin-$\downarrow$ fermion at the same site. Under the JW transformation, the fermionic operators are mapped as:
\begin{align}
c_{i}^\dagger &= \Big(\bigotimes_{j < i} \sigma_j^{z}\Big)\,\frac{\sigma_i^{x} - i\,\sigma_i^{y}}{2},\\
c_{i} &= \Big(\bigotimes_{j < i} \sigma_j^{z}\Big)\,\frac{\sigma_i^{x} + i\,\sigma_i^{y}}{2},\\
n_{i s} &= c^{\dagger}_{i s} c_{i s} = \frac{1 - \sigma_i^{z}}{2},
\end{align}
where $\sigma_j^{x}, \sigma_j^{y}, \sigma_j^{z}$ are Pauli matrices acting on qubit $j$. The corresponding qubit Hamiltonian turns out to be
\begin{equation}
\label{eq:JW_Hubbard_qubit}
\begin{aligned}
H
&= -\frac{t}{2}\sum_{i=1}^{N-1}\!\Big(
  \sigma_i^{x}\sigma_{i+1}^{x} + \sigma_i^{y}\sigma_{i+1}^{y}
+ \sigma_{i+N}^{x}\sigma_{i+1+N}^{x}  \\&\quad +\sigma_{i+N}^{y}\sigma_{i+1+N}^{y}
\Big) + \frac{U}{4}\sum_{i=1}^{N} \sigma_i^{z}\sigma_{i+N}^{z}\, .
\end{aligned}
\end{equation}
We choose HVA as the variational ansatz for this model. Specifically, each layer applies $\mathcal{U}_l(\bm{\theta}) {=} e^{-i \theta^{U}_{l} H_{U}} e^{-i \theta^{(1)}_{l} H_{t}^{(1)}} e^{-i \theta^{(2)}_{l} H_{t}^{(2)}}$ with on-site interaction $H_{U}{=}U\sum_{i}(n_{i\uparrow}{-}\tfrac12)(n_{i\downarrow}{-}\tfrac12)$ and hopping interactions $H_{t}^{(1)/(2)}{=}\sum_{i\in \mathrm{odd/even},\,s}(c^{\dagger}_{i s}c_{i+1,s}+\mathrm{H.c.})$.
We implement them via $U_{ZZ}(\theta){=}e^{-i(\theta/4)\sigma^{z}_{i}\sigma^{z}_{i+N}}$ and 
$U_{XY}(\theta){=}e^{-i \theta/2 (\sigma^{x}_{i}\sigma^{x}_{i+1}+\sigma^{y}_{i}\sigma^{y}_{i+1})}$. The circuit is shown in Fig.~\ref{fig:ansatz}(e).

\begin{table*}[t]
\caption{Quantum and classical resources required to reach $F=0.99$ across three models. $N_I$ denotes the number of optimization iterations, $C_R$ denotes the cumulative optimization cost.}
\label{tab:resources}
\begin{ruledtabular}
\renewcommand{\arraystretch}{1.5}
\begin{tabular}{l l c c c c c c c c}
      &               & \multicolumn{2}{c}{\textbf{1D Ising}} 
                       & \multicolumn{2}{c}{\textbf{2D Ising}} 
                       & \multicolumn{2}{c}{\textbf{Cluster-Ising}} 
                       & \multicolumn{2}{c}{\textbf{Fermi-Hubbard}} \\
      &               & Our Method & HEA 
                       & Our Method & HVA 
                       & Our Method & HVA 
                       & Our Method & HVA \\
\colrule
\multirow{2}{*}{\shortstack[l]{Quantum\\Resources}}
      & Layers        & 8          & 12 
                       & 6          & 8 
                       & 6          & 9 
                       & 5          & 9 \\
      & 2-Qubit Gates & $\sim 112$ & 144 
                       & $\sim 300$ & 432 
                       & $\sim 450$ & 558 
                       & $\sim 128$ & 252 \\
\colrule
\multirow{2}{*}{\shortstack[l]{Classical\\Resources}}
      & $N_I$         & 1100       & 1500 
                       & 160        & 120 
                       & 270        & 400 
                       & 500        & 500 \\
      & $C_R$         & 189600     & 432000 
                       & 30800      & 38880 
                       & 54550      & 118800 
                       & 21600      & 45000 \\
\end{tabular}
\end{ruledtabular}
\end{table*}

The choice of the input state critically affects VQAs, especially HVA, where optimization is sensitive to the structure of the reference states. For the Fermi-Hubbard model, common input-state choices include the free-fermion ground state (i.e., the $U=0$ solution) and mean-field Hartree-Fock states~\cite{Romero_2019, Bartlett1989133, Taube2006New}. While both approaches have been widely used, they present distinct trade-offs in terms of symmetry, expressivity, and circuit cost~\cite{PhysRevResearch.4.023190}. Note that the Fermi-Hubbard Hamiltonian in Eq.~(\ref{eq:hubbard}) conserves the number of particles with a given spin, namely $[H,n^{\rm tot}_{\uparrow}]{=}[H,n^{\rm tot}_{\downarrow}]{=}0$, where $n^{\rm tot}_{\uparrow,\downarrow}=\sum_i n_{i\uparrow,\downarrow}$. Therefore, one can use a quantum circuit that preserves such symmetry to enhance the performance of the VQE~\cite{Lyu2023symmetryenhanced}. In this case, if the input state has a given number of particles with spin $\uparrow$ (or spin $\downarrow$) then the output of the circuit also automatically remains in such subspace. In spin language, it translates to $[H,S_z]=0$, where $S_z=\tfrac{1}{2}\sum_i \sigma_i^z$ and the input state has to have an equal number of $\ket{0}$ and $\ket{1}$ for an even system size $N$. 
For instance, for a four-site half-filled chain, which is mapped to $8$ qubits, we use the initial state $\ket{11000011}$ to start the algorithm. Building on this reference state, our method constructs a superposition of selected low-energy basis states using an encoder circuit.

We evaluate this input-state design strategy on a 4-site Fermi-Hubbard model (Fig.~\ref{Hubbard_site}) with open boundary conditions at half-filling for $U = 2$, $5$, and $10$. In each run we sample \(M=70\) computational-basis states within the fixed-particle-number sector (4 fermions), select $m=4$, and build the encoder from them. For each circuit depth we average fidelity over $100$ random initializations. The results are shown in Fig.~\ref{Hubbard}. In Figs.~\ref{Hubbard}(a)-(b), we plot the obtainable average energy as well as the fidelity versus the number of circuit layers for the case of $U = 2$ (i.e., weakly correlated regime). The bitstring state provides a good starting point: a depth-$9$ HVA circuit initialized from it achieves a fidelity of $0.99$. In contrast, our method, which incorporates an encoder layer, reaches fidelity $0.99$ using only $5$ layers. As $U$ increases, the enhanced performance of the input-state design algorithm becomes more pronounced. This is because by increasing $U$ the energy gap between the two lowest eigenstates becomes very small and low-lying excitations proliferate, making the optimization highly sensitive to initialization. As shown in Figs.~\ref{Hubbard}(c)-(d), for the case of $U=5$ (i.e., medium interaction limit), a single bitstring input faces barren-plateau-like behavior, with fidelity saturating around $0.6$. Introducing our encoder layer markedly improves performance at the same circuit depth, pushing the fidelity above $0.99$. In Figs.~\ref{Hubbard}(e)-(f) we repeat the analysis for the case of $U=10$ (i.e., strong interaction limit), in which again our performance overcomes the training issue. In particular, while the conventional method only achieves fidelity $0.6$, our input-state design can reach fidelity of $0.99$. This clearly shows that the encoder enhances the expressivity and physical relevance of the initialization, thereby improving the overall performance of the variational ansatz.

To summarize our results for the three examples, in Table~\ref{tab:resources}, we compare both quantum and classical resources between conventional approaches and our method. As the data clearly demonstrate, our input-state design approach outperforms the conventional algorithms in terms of both quantum and classical resources.

\section{Conclusion}\label{conclusion}

In this work, we have developed input-state design as a powerful complement to circuit design for VQAs. By introducing a measurement-based encoder, we prepare an input state whose reachable set under a fixed-depth circuit is more likely to include the target state than the conventional approaches. Theoretically, we prove that by constructing a carefully designed input state for the ansatz $U(\boldsymbol{\theta})$ as a superposition of mutually orthogonal candidate states ${\ket{\psi_j}}$, one can systematically improve the performance of the given VQA. Empirically, across a range of models---including the one-dimensional Ising model, the two-dimensional Ising model, the cluster-Ising model, and the strongly correlated Fermi--Hubbard model---our method consistently outperforms HEA/HVA baselines under the same gate budget. A key advantage of this approach is its universality: it integrates seamlessly with any variational ansatz architecture and can be readily combined with existing techniques for noise mitigation, parameter initialization, and classical optimization. These features position input-state design as a practical tool for improving quantum state preparation and ground-state energy estimation on near-term hardware, particularly in regimes where circuit depth and quantum resources are limited.

\begin{acknowledgments}
This work was supported by the National Natural Science Foundation of China (Grants No.~92265208, No.~12274059, No.~12574528 and No.~1251101297) and the Sichuan Science and Technology Program (Grant No.~2025YFHZ0336). 
\end{acknowledgments}

\appendix*

\section{Circuit design of the low-depth encoder}\label{V_exmple}

In this section, we present a concrete example demonstrating how the proposed input-state design is realized in practice through the construction of the encoding layer $V(\boldsymbol{\gamma})$. Specifically, we design $V(\boldsymbol{\gamma})$ to prepare a superposed input state to solve the ground energy of a $12$-qubit one-dimensional Ising model. We begin by training a $5$-layer HEA as the baseline circuit. Following Algorithm~\ref{alg:enhanced_vqe}, six computational-basis states are then selected—$|0\rangle$, $|30\rangle$, $|60\rangle$, $|480\rangle$, $|960\rangle$, and $|2049\rangle$---which together span a low-energy subspace identified by the pre-optimized HEA. The encoding layer $V(\boldsymbol{\gamma})$ is constructed to combine these basis states into a superposition that serves as the new input-state, $\ket{\Psi_0(\boldsymbol{\gamma})}=V(\boldsymbol{\gamma})\,|0\rangle$.

The encoder $V(\boldsymbol{\gamma})$ is composed of four components $U_1, U_2, U_3,$ and $U_4$:
\begin{align*}
U_1 &= CR_y^{(4,3)} \cdot CR_z^{(4,3)} \cdot CR_y^{(4,7)} \cdot CR_z^{(4,7)}\\
 &\quad \cdot CR_y^{(10,11)} \cdot CR_z^{(10,11)} \cdot CR_y^{(10,6)} \cdot CR_z^{(10,6)}\\
 &\quad \cdot CR_y^{(10,7)} \cdot CR_z^{(10,7)}\\
U_2 &= CCR_y^{(\bar{4},\bar{10},1)} \cdot CCR_z^{(\bar{4},\bar{10},1)} \cdot CCR_y^{(\bar{4},\bar{10},12)} \cdot CCR_z^{(\bar{4},\bar{10},12)}\\
U_3 &= CR_y^{(6,8)} \cdot CR_z^{(6,8)} \cdot CR_y^{(6,9)} \cdot CR_z^{(6,9)}\\
 &\quad \cdot CR_y^{(6,10)} \cdot CR_z^{(6,10)} \cdot CR_y^{(6,4)} \cdot CR_z^{(6,4)}\\
 &\quad \cdot CR_y^{(6,5)} \cdot CR_z^{(6,5)}\\
 U_4 &= R_y^{(6)} 
\end{align*}
The complete encoding layer is then $V(\boldsymbol{\gamma}) = U_1\, U_2\, U_3\, U_4$. Here, $CR_\alpha^{(c,t)}$ denotes a controlled rotation $R_\alpha$ on target qubit $t$ conditioned on control qubit $c$, and $CCR_\alpha^{(c_1,c_2,t)}$ denotes a control-control rotation on target $t$ conditioned on controls $c_1$ and $c_2$. To execute the encoder on NISQ hardware, the three-qubit gates are decomposed into one- and two-qubit operations. Following Ref.~\cite{Nielsen_Chuang_2010}, a control-control-unitary can be implemented using at least five basic gates. In total, constructing $V(\boldsymbol{\gamma})$ for this example requires approximately $41$ elementary gates—comparable to the $60$ gates used by a single HEA layer. This shows that the number of gates required to construct $V(\bm\gamma)$ is approximately the same as that for one layer of the ansatz, indicating that the comparison is fair. Moreover, the circuit complexity of $V(\boldsymbol{\gamma})$ scales linearly with the number of qubits, confirming that the encoder remains a low-depth circuit.

\bibliographystyle{apsrev4-2}
\bibliography{ref}

\clearpage

\end{document}